\documentclass{article}
\usepackage{graphicx} 
\usepackage[utf8]{inputenc}
\usepackage{amsmath,amsbsy,amssymb,array}
\usepackage{caption}
\usepackage{multirow,multicol,tabularx,booktabs}
\usepackage{graphicx,placeins,color,url}
\usepackage{amsthm}
\usepackage[a4paper]{geometry}
\usepackage{authblk}
\usepackage{algorithm} 
\usepackage{algpseudocode} 
\usepackage{comment}
\usepackage[title]{appendix}
\RequirePackage[numbers]{natbib}
\usepackage{adjustbox}
\usepackage{graphicx}
\usepackage{amsmath}
\newtheorem{thm}{Theorem}

\newtheorem{prop}[thm]{\bf Proposition}

\newtheorem{defi}{Definition}
\newtheorem{rem}{Remark}

\def\Eb{{\mathbb E}}

\def\Rb{\mbox{$\mathbb R$}}

\title{Asset and Factor Risk Budgeting: A Balanced Approach}
\author{Adil Rengim CETINGOZ\thanks{Université Paris 1 Panthéon-Sorbonne, Centre d’Economie de la Sorbonne, 106 Boulevard de l’Hôpital, 75642 Paris Cedex 13, France, adil-rengim.cetingoz@etu.univ-paris1.fr.\textit{ Corresponding author.}} \quad Olivier GU\'EANT\thanks{Université Paris 1 Panthéon-Sorbonne, Centre d’Economie de la Sorbonne, 106 Boulevard de l’Hôpital, 75642 Paris Cedex 13, France, olivier.gueant@univ-paris1.fr.}}      
\date{}

\begin{document}

\maketitle
\begin{abstract}
Portfolio optimization methods have evolved significantly since Markowitz introduced the mean-variance framework in 1952. While the theoretical appeal of this approach is undeniable, its practical implementation poses important challenges, primarily revolving around the intricate task of estimating expected returns. As a result, practitioners and scholars have explored alternative methods that prioritize risk management and diversification. One such approach is Risk Budgeting, where portfolio risk is allocated among assets according to predefined risk budgets. The effectiveness of Risk Budgeting in achieving true diversification can, however, be questioned, given that asset returns are often influenced by a small number of risk factors. From this perspective, one question arises: is it possible to allocate risk at the factor level using the Risk Budgeting approach? First, we introduce a comprehensive framework to address this question by introducing risk measures directly associated with risk factor exposures and demonstrating the desirable mathematical properties of these risk measures, making them suitable for optimization. Then, we propose a novel framework to find portfolios that effectively balance the risk contributions from both assets and factors. Leveraging standard stochastic algorithms, our framework enables the use of a wide range of risk measures to construct diversified portfolios.

\end{abstract}
\vfill
\begin{flushleft}
\end{flushleft}
\newpage
\section{Introduction}
In Markowitz' seminal 1952 paper~\cite{markowitz1952}, portfolio selection is described as a two-stage procedure. The first stage involves forming beliefs about the future evolution of asset returns, and the second stage concerns portfolio construction based on these beliefs. Modern portfolio theory provides fundamental results for the second stage, by mathematically formalizing the benefits of the ubiquitous concept of diversification and proposing a clear investment objective for investors, based on mean-variance analysis. However, the complexity of the first stage often leads to erroneous beliefs on expected returns that Markowitz' portfolios excessively reflect in asset weightings (see~\cite{best1991sensitivity} and \cite{michaud1989markowitz}). 

\medskip
Two main lines of research have therefore coexisted to improve the portfolio selection process in the academic literature and in the asset management industry.

\medskip
On the one hand, a considerable amount of effort has been devoted to improving the second stage, by reducing the sensitivity to input parameters and obtaining practically more plausible portfolios. Some methods have attempted to retain expected returns within the picture, such as the Black-Litterman model \cite{black1992global} and entropy pooling \cite{meucci2010fully}, while other methods omit them fully (or partially) and only involve the covariance matrix in the portfolio selection process -- so-called risk-based methods. Among risk-based methods, there are several fundamental approaches based on different rationales.\footnote{Interestingly, most of them boil down mathematically to minimizing the portfolio risk under a specific constraint with different parameters (see \cite{gava2022properties} and \cite{richard2015smart}).} Minimum Variance (MV) portfolio is one of them, corresponding to the portfolio with the minimum risk (as measured by variance) among all the possible fully invested portfolios. Another approach is that of the Most-Diversified portfolio, as defined in \cite{choueifaty2008toward}, which seeks to maximize the ratio between the weighted sum of the volatilities of individual assets and the total volatility of the portfolio. Equal Risk Contribution (or Risk Parity) is another risk-based approach that differs from the previous two in motivation. Rather than focusing on optimizing a final statistic of the portfolio, it aims to allocate portfolio risk equally across its assets. Risk Budgeting (RB) is a generalization of this approach, which allows the user to choose the target risk contribution of each asset by setting asset-specific risk budgets. Risk Budgeting has become a popular method, particularly over the last two decades, thanks to some interesting theoretical properties (see \cite{bruder2016risk}, \cite{bruder2012managing} and \cite{roncalli2013introduction}). 

\medskip
On the other hand, with regard to the first stage of portfolio selection, the quest for a better understanding of the true drivers of asset returns and the risk-return relationship dates back to the 1960s, when the Capital Asset Pricing Model (CAPM) was independently developed by Sharpe, Lintner, Mossin and Treynor (see \cite{lintner1965security, mossin1966equilibrium, sharpe1964capital, Treynor1962JackT}). The CAPM establishes a linear relationship between the expected excess returns of risky assets and the expected excess return of the market portfolio, resulting in a single-factor model for asset returns. The arbitrage pricing theory (APT) developed in \cite{ross1976arbitrage} extends this original approach to a multi-factor framework, without identifying the risk factors. The Fama-French three-factor model proposed in \cite{fama1992cross} is one of the most famous factor models, expanding the CAPM by introducing size and value factors on top of the market factor. There are other prominent models that later introduced new factors like momentum, profitability, quality, investments, etc. (see for instance \cite{carhart1997persistence} and \cite{fama2015five}). Over the years, the number of factors proposed by academics and practitioners has grown considerably, giving rise to the phenomenon of the \textit{factor zoo}, as stated by \cite{cochrane2011presidential}. However, a recent study argues that most of these factors are redundant (see~\cite{swade2023factor}).  

\medskip
In the context of a factor model, exploring and controlling exposure to underlying factors become crucial. For example,  \cite{greenberg2016factors} explores portfolios designed to achieve specific factor exposures, highlighting the existence of infinitely many portfolios meeting the desired factor exposures. In~\cite{dichtl2019optimal}, the effectiveness of the equally-weighted factor portfolio is compared to alternative factor-allocation schemes. Although risk-based approaches have been explored, in particular the feasibility of Risk Budgeting to obtain diversified factor-focused portfolios, it should be noted that, with the exception of a few papers dealing with factor risk contributions in the context of correlated factors (see \cite{meucci2007risk} and \cite{roncalli2016risk}), most papers consider factor risk budgeting portfolios on uncorrelated factors. This is often achieved by using techniques like principal / independent component analysis or Minimum-Torsion Bets to decorrelate existing factors (see \cite{amato2020diversifying}, \cite{lassance2022optimal}~and~\cite{meucci2015risk}). 

\medskip
The primary objective of this paper is to propose a rigorous mathematical framework for constructing diversified portfolios that effectively mitigate the concentration of risk at both the asset and factor levels, especially in the presence of dependent or correlated risk factors. In a setting where asset returns are driven by specific factors, the central challenge for portfolio managers is to manage exposures to these factors effectively. Our first contribution addresses this challenge by introducing a specific formalization through a factor risk measure -- a mapping from factor exposures to risk -- that enables the selection of a unique vector of asset exposures. We demonstrate the financial viability and essential mathematical properties of the proposed factor risk measure, which are crucial for its effective use in risk management and portfolio optimization. These findings and ideas might be of interest beyond risk budgeting techniques for factor investing in general.

\medskip
Additionally, we show that our factor risk measure leads to a formal definition and solution of the Factor Risk Budgeting problem, which involves allocating portfolio risk across chosen risk factors according to predefined risk budgets for each factor. This represents the second contribution of our work. Furthermore, we present an optimization framework that integrates the principles of Factor Risk Budgeting with classical Risk Budgeting to design portfolios that account for both asset and factor risk contributions. To the best of our knowledge, this integrated framework represents a novel contribution to the asset management literature and aims to serve as a practical methodology for achieving diversification. Finally, we highlight the stochastic formulation of the proposed optimization problem, which allows portfolio managers to efficiently work with a wide range of risk measures when quantifying portfolio risk.  

\medskip
In Section~\ref{intro_RB}, we outline the classical Risk Budgeting problem and present an efficient method to find Risk Budgeting portfolios. In Section~\ref{section3}, we present linear factor models and the concept of factor exposures. Then, we introduce factor risk measures which map factor exposures to risk. We show that these factor risk measures possess all the necessary mathematical properties to be used for risk management and portfolio optimization, especially for Risk Budgeting. Section~\ref{section4} precisely defines the Factor Risk Budgeting problem. In Section~\ref{asset_factor_setting}, we propose the Asset-Factor Risk Budgeting framework that seeks a balance between classical Risk Budgeting and Factor Risk Budgeting, considering both asset and factor risk contributions. Section~\ref{section6} presents two case studies where we construct equity portfolios under the BARRA risk model and cross-asset portfolios in the context of a macroeconomic factor model to exemplify the use of Factor Risk Budgeting and Asset-Factor Risk Budgeting portfolios.

\section{Risk Budgeting: the classical problem}
\label{intro_RB}

\subsection{Notations and main concepts}
In this paper, we consider a probability space $(\Omega,\mathcal{F}, \mathbb{P})$ where we denote by $L^0(\Omega,\mathbb{R}^d)$ the set of  $\mathbb R^d$-valued random variables, i.e. measurable functions defined on $\Omega$ with values in $\mathbb R^d$ ($d \ge 1$).

\medskip
At the core of risk management lies the fundamental concept of risk measures: mathematical functions that map random variables onto real numbers. In the context of Risk Budgeting, our focus centers on risk measures $\rho$ defined on $ A \subset L^0(\Omega,\mathbb{R})$ satisfying the following properties:\footnote{The set $A$ depends on the context. It is typically $L^1(\Omega,\mathbb{R})$ or $L^2(\Omega,\mathbb{R})$.}
\begin{eqnarray*}
 \forall Z \in  A , \forall \lambda\geq0 , \quad \rho(\lambda Z) = \lambda \rho(Z) & \text{(positive homogeneity)}\\
 \forall Z_1,Z_2  \in  A , \quad \rho(Z_1+Z_2) \leq \rho(Z_1)+\rho(Z_2) & \text{(sub-additivity)}.
\end{eqnarray*}

In asset management, the risk is associated with a portfolio made of multiple assets. Therefore, we consider a financial universe of $d$ assets whose joint returns are given by a random vector $X = (X_1, \ldots, X_d)'$ where $(1, X_1, \ldots, X_d)$ is assumed to be linearly independent. Given a vector of asset exposures denoted by $y \in \mathbb{R}^d$, the loss associated with $y$ is the random variable $-y'X$ and the risk is $\rho(-y'X)$.\footnote{The vector $y$ can be interpreted as a portfolio where short positions and leverage are allowed.} This leads to a function $\mathcal{R}_{\rho,X}$ that we call an asset exposure risk measure associated with $\rho$ and $X$:
$$\mathcal{R}_{\rho, X}: y \mapsto \rho(-y'X).$$
$\mathcal{R}_{\rho, X}$  is hereafter denoted by $\mathcal{R}$ when there is no ambiguity. We assume that $\mathcal{R}$ is continuous on~$\mathbb{R}^d$ and continuously differentiable on $\mathbb{R}^d\setminus\{\mathbf{0}\}$.\footnote{The risk measure $\rho$ and the distribution of $X$ are assumed to be chosen to ensure that the mentioned properties for $\mathcal{R}$ are satisfied.} We also assume that $\forall y \in \mathbb{R}^d\setminus\{\mathbf{0}\}$, $\mathcal{R}(y)>0$, i.e. the risk of a non-zero financial position is always positive. We finally make the assumption that $\mathcal{R}$ is additive solely under scaled asset exposures:
$$\forall y_1,y_2  \in   \mathbb{R}^d\setminus\{\mathbf{0}\}, \quad \mathcal{R}(y_1+y_2) = \mathcal{R}(y_1)+\mathcal{R}(y_2) \implies \exists c > 0, \quad y_1=c y_2.$$

In practice, the vectors of asset exposures $y$ of portfolios are often vectors of weights. In this paper, a portfolio is denoted by a vector of weights $\theta$ which might belong to $$\Delta_d = \{\theta \in \mathbb{R}^d | \theta_1 + \ldots +\theta_d = 1\} \qquad \text{for general portfolios},$$$$\Delta^{\geq 0}_d = \{ \theta \in \mathbb{R}_+^d | \theta_1 + \ldots +\theta_d = 1\} \qquad \text{for long-only portfolios},$$ or $$\Delta^{>0}_d = \{ \theta \in (\mathbb{R}^*_+)^d | \theta_1 + \ldots +\theta_d = 1\} \qquad \text{for strictly long-only portfolios.}$$ 

\medskip
While $\mathcal{R}(\theta)$ provides an assessment of the total risk associated with the portfolio, it is equally important to measure the individual contributions of each asset to this total risk. In other words, it is interesting to understand how $\mathcal{R}(\theta)$ can be decomposed into asset-level risk contributions. This decomposition of risk can be performed (thanks to the positive homogeneity property), using Euler homogeneous function theorem, i.e.
$$\mathcal{R}(\theta) = \sum_{i=1}^d \theta_i \partial_i \mathcal{R}(\theta),$$
where $\partial_i \mathcal{R}(\theta)$ is the partial derivative of $\mathcal{R}$ with respect to the variable~$\theta_i$ at point $\theta$. In the above sum, the $i^{\text{th}}$ term is called the risk contribution of asset $i$ to the total risk of the portfolio. The concept of risk contribution is a fundamental notion in the Risk Budgeting framework. 

\subsection{Formulation of the Risk Budgeting problem and associated theoretical results}

In the context of Risk Budgeting, the central objective is to achieve predefined levels of risk contributions from the assets of the portfolio. Formally, this task consists of determining portfolio weights to ensure that the proportion of risk contribution allocated to each asset, relative to total risk, corresponds precisely to the risk budgets assigned to them.

\begin{defi}
\label{original_def_rb}
Let $b \in \Delta^{>0}_d$ be a vector of risk budgets. We say that a vector of weights $\theta \in \Delta^{>0}_d$ solves the Risk Budgeting problem $\text{RB}_{b}$ if and only if 
$$\quad \theta_i \partial_i \mathcal{R}(\theta) = b_i \mathcal{R}(\theta),$$ 
for every $i \in \{1,\ldots,d\}$.
\end{defi}
\begin{rem}
This is the classical definition of Risk Budgeting where only strictly long-only portfolios are taken into account. If we remove this constraint, there exist multiple solutions in $\Delta_d$ to $\text{RB}_{b}$ (see \cite{bai2016least} for a detailed analysis in the case of volatility).
\end{rem}
Questions that naturally arise from the above definition concern the existence and uniqueness of a solution. The constructive theorem that follows, proved in \cite{cetingoz2023risk}, states the existence of a unique solution to the Risk Budgeting problem.
\begin{thm}
\label{thm_existence}
Let $b \in \Delta^{>0}_d$.

Let $g : \mathbb{R_+} \rightarrow \mathbb{R}$ be a continuously differentiable convex and increasing function.
Let the function $\Gamma_g:(\Rb_+^*)^d \rightarrow \Rb$ be defined by 
$$ \Gamma_g: y \mapsto g\big(\mathcal{R}(y)\big) - \sum_{i=1}^d b_i \log{y_i}.$$
There exists a unique minimizer $y^*$ of the function $\Gamma_g$.

There exists a unique solution $\theta \in \Delta^{>0}_d$ of $\text{RB}_b$ and $\theta = \frac{y^*}{\sum_{i=1}^d y^*_i}$.
\end{thm}

The importance of Theorem~\ref{thm_existence} goes beyond confirming the existence of a unique solution. Rather than directly engaging with the system of non-linear equations that characterize the original problem given by Definition~\ref{original_def_rb}, this theorem highlights the possibility of approaching the Risk Budgeting problem through unconstrained convex minimization problems. The latter technique has clear benefits in terms of mathematical properties and computational efficiency. Furthermore, the choice of the function $g$ is free, making the approach versatile, as we illustrate in the next section.

\subsection{Computation of Risk Budgeting portfolios}
Early research on Risk Budgeting and most applications have focused on the case where the co-movements of asset returns are measured by their covariance matrix $\Sigma$ and the risk of a portfolio with asset exposures $y$ is measured by its volatility, i.e. $\mathcal{R}(y) := \sqrt{y'\Sigma y}$ ($\mathcal R$ satisfies our assumptions whenever $\Sigma$ is invertible -- an assumption that we consider throughout the paper). For a given $\Sigma$, the volatility of the portfolio and its gradient can easily be computed. Consequently, the Risk Budgeting problem for volatility can be solved efficiently using a gradient descent procedure, by choosing $g(x) = x^2$ and minimizing over $\mathbb{(R_+^*)}^d$ the function  
\begin{equation}
\label{rb_problem_vol}
 \Gamma_{g} : y \mapsto \big(\mathcal{R}(y)\big)^2 - \sum_{i=1}^d b_i \log{y_i} = y'\Sigma y - \sum_{i=1}^d b_i \log{y_i}. 
\end{equation}
However, the scope of risk measures used in portfolio management goes well beyond volatility and includes a wide variety of symmetric and asymmetric risk measures. In \cite{ararat2024mad}, the authors consider Risk Parity when the risk measure is chosen to be mean absolute deviation around the mean and provide a comprehensive mathematical analysis of the specific problem. Considering the inherent asymmetry and heavy-tailed nature of asset return distributions, the literature on Risk Budgeting has also touched upon other asymmetric risk measures. For instance, in~\cite{lezmi2018portfolio}, Risk Budgeting portfolios are constructed for Expected Shortfall.\footnote{The Expected Shortfall at level $\alpha \in (0,1)$ of a continuous random variable (loss) $Z \in L^1$ with positive density is defined by $\mathbb{E}[Z | Z \geq \text{VaR}_\alpha(Z)]$ where $\mathbb{P}(Z \geq \text{VaR}_\alpha(Z)) = 1-\alpha$.} In order to employ the aforementioned gradient descent procedure effectively, they make a specific assumption about asset returns. Precisely, they assume that asset returns follow a mixture of two multivariate Gaussian distributions -- a setting in which Expected Shortfall exhibits a semi-analytic formula (see~\cite{cetingoz2023risk} for an extension to Student mixtures). In cases where there are no parametric assumptions on asset returns or when dealing with more complex risk measures, gradient computation often necessitates approximations through computationally intensive methods (see \cite{gava2021turning}).\footnote{In \cite{jaimungal2023risk}, the authors tackle the Risk Budgeting problem for dynamic risk measures.}     

\medskip
There exists, however, an alternative path that relies on a common representation shared by a wide range of risk measures. Precisely, we are interested in RB-compatible risk measures as defined in \cite{cetingoz2023risk} which are characterized as minima. Indeed, in that case, we can transform the optimization problem, whose solution (after normalization) gives the Risk Budgeting portfolio, into a stochastic optimization problem.

\medskip
Formally, if there exists a continuously differentiable, convex and increasing function $g : \mathbb{R_+} \rightarrow \mathbb{R}$ such that 
\begin{equation*}
g\big(\mathcal{R}(y)\big) = g\big(\rho(-y'X)\big) = \min_{\zeta \in \mathcal{Z}} \mathbb E\big[H(\zeta, -y'X) \big]
\end{equation*} 
for some set $\mathcal{Z}$ and some function $H$, then $\Gamma_{g}$ in Theorem~\ref{thm_existence} writes 

$$\Gamma_{g} : y \mapsto \min_{\zeta \in \mathcal{Z}} \mathbb E\big[H(\zeta, -y'X) \big] - \sum_{i=1}^d b_i \log{y_i}.
$$

Therefore, the $\text{RB}_b$ problem boils down to the stochastic optimization problem

\begin{equation*}
  \min_{y \in (\mathbb{R}_+^*)^d, \zeta \in \mathcal Z} \mathbb E\Big[H(\zeta, -y'X) - \sum_{i=1}^d b_i \log{y_i}\Big].
\end{equation*}
If $(y^*,\zeta^*)$ is a solution of the above stochastic optimization problem, then $\theta := y^*/\sum_{i=1}^d y^*_i$ solves $\text{RB}_b$.

\medskip
In fact, the set of risk measures that enable the above stochastic path for the computation of Risk Budgeting portfolios is quite broad. All deviation measures and spectral risk measures (as formalized in \cite{cetingoz2023risk}) fall within this framework. Volatility and Expected Shortfall are the two most popular members of these classes, respectively. For volatility, i.e. $\rho(Z) = \min_{\zeta \in \mathbb R} \sqrt{\Eb\left[(Z-\zeta)^2\right]}$, the $\text{RB}_b$ problem becomes,\footnote{Of course, using a stochastic gradient descent algorithm in the case of volatility is not the best option. A gradient descent on Eq. \eqref{rb_problem_vol} is more efficient.} upon choosing $g: x \mapsto x^2$,
\begin{equation*}
\label{stocRB_vol}
    \min_{(y,\zeta) \in (\mathbb R_+^*)^d\times\mathbb R} \mathbb E\Big[(-y'X-\zeta)^2 - \sum_{i=1}^d b_i \log{y_i}\Big].
\end{equation*}
For Expected Shortfall, i.e. $\rho(Z)=\min_{\zeta \in \mathbb R} \mathbb E\big[ \zeta + \frac 1{1-\alpha} (Z-\zeta)_+\big]$ as formulated by Rockafellar-Uryasev (see \cite{rockafellar2002conditional}), the $\text{RB}_b$ problem, when $g$ is the identity map, becomes 
\begin{equation}
\label{stocRB_ES}
    \min_{(y,\zeta) \in (\mathbb R_+^*)^d\times\mathbb R} \mathbb E\Big[ \zeta + \frac 1{1-\alpha} (-y'X-\zeta)_+ - \sum_{i=1}^d b_i \log{y_i}\Big].
\end{equation}

It is essential to highlight that this approach is suited to a broad spectrum of asset return distributions. Furthermore, it exhibits a high level of versatility, when it comes to using different risk measures. For many risk measures, a stochastic gradient descent algorithm based on the above ideas is a great option to effectively tackle Risk Budgeting problems.\footnote{In \cite{da2023risk}, the authors propose using cutting planes algorithm for the specific case of Expected Shortfall.} Beyond volatility and Expected Shortfall, some other examples of risk measures covered by the stochastic framework are mean absolute deviation around the median, power spectral risk measure and variantile \cite{wang2020risk}. 

\medskip
In this section, we have outlined the Risk Budgeting problem and presented a computational approach for obtaining Risk Budgeting portfolios. These portfolios are designed to provide optimal asset allocation, achieving predefined levels of risk contribution from each asset. However, in a financial context where asset returns are influenced by a small number of common underlying factors, it becomes crucial to consider controlling exposure to these factors rather than focusing solely on asset weights. Consequently, we should shift our focus towards evaluating the risk contributions of these underlying factors, as opposed to those of individual assets. These critical considerations are the topic of the next sections.

\section{Assessing risk through the lens of risk factors}
\label{section3}
\subsection{Linear factor models and factor exposures}
It is common practice for portfolio and risk managers to write $X$ as a linear function of $m$ factor returns modelled by $F \in L^0(\Omega,\mathbb{R}^m)$ and $d$ idiosyncratic components modelled by $\epsilon \in L^0(\Omega,\mathbb{R}^d)$ where $F$ and $\epsilon$ are assumed to be independent. This modelling framework is commonly known as a linear factor model and is mathematically expressed as
\begin{equation}
\label{factor_model}
  X = \beta F + \epsilon,  
\end{equation}
where $\beta \in \mathbb R^{d \times m}$ is the matrix of factor loadings.\medskip

Factor loadings characterize the sensitivity of assets to specific risk factors. They provide a quantitative measure of how much the return of an asset are expected to fluctuate in response to a one-unit change in a particular risk factor. These risk factors encompass a broad spectrum, including economic factors such as economic growth and inflation, stock-specific features like size, value, quality and profitability, and statistical factors extracted from data. Factor loadings are typically estimated through statistical techniques like regression analysis.

\medskip
The number of factors $m$ is typically chosen to be much smaller than the total number of portfolio assets $d$. Throughout this paper, we maintain the condition $m<d$ and we also assume that the factor loading matrix $\beta$ is full rank, i.e. $\text{rank}(\beta)=m$.

\medskip

In the modeling framework of Eq.~(\ref{factor_model}), the portfolio loss for a given vector of asset exposures~$y$ is 
$$-y'X = -w'F - y'\epsilon,$$
where $w=\beta'y$ is the vector of factor exposures associated with $y$.
\medskip

In such a model, one is mostly interested in managing the portfolio risk through factor exposures and achieving diversification among factors rather than solely focusing on the exposures to the individual assets since their returns are actually dependent on the returns of the underlying factors. Therefore, it makes sense to try to achieve a desired level of factor exposures $w$ by finding a vector of asset exposures $y$ that satisfies $w=\beta'y$. However, an issue arises due to the existence of an infinite number of asset exposure vectors that meet this condition.

\medskip

Because $\beta$ is full rank, one of these infinitely many vectors is $\bar{y}_w = \beta(\beta'\beta)^{-1}w$.\footnote{For instance, \cite{roncalli2016risk} considers this specific case.} Therefore, we clearly see that $\beta'y=w \iff \beta'y=\beta'\bar{y}_w \iff y-\bar{y}_w \in \text{Ker}(\beta') = \text{Im}(\beta)^{\perp}$. Thus, the set of asset exposures with the factor exposures $w$ is the affine subspace $\bar{y}_w+\text{Im}(\beta)^{\perp}$. Amongst all the vectors $y \in \bar{y}_w+\text{Im}(\beta)^{\perp}$ of asset exposures compatible with a vector of factor exposures $w$, some may be more favourable than others when considering additional criteria that are important for portfolio managers. In the following section, we address this concern and introduce a risk-based framework to tackle this issue.

\subsection{From factor exposures to portfolio risk: a factor risk measure}
Within the set of asset exposures that achieve the desired level of factor exposures, it can indeed be interesting to find the vector of asset exposures associated with the minimum level of risk. This is the purpose of the next paragraphs. We introduce below a function that maps a vector of factor exposures into the minimum achievable risk, among all feasible asset exposures.

\begin{defi}
\label{def_new_rm}
For a full-rank matrix $\beta \in \mathbb R^{d \times m}$ and an asset exposure risk measure $\mathcal{R}:\mathbb{R}^d \rightarrow\mathbb{R}$ satisfying the assumptions of Section 2, we define the factor risk measure $\mathcal{S}_{\mathcal R, \beta}:\mathbb{R}^m \rightarrow \mathbb{R}$ (hereafter simply denoted by $\mathcal{S}$) by
$$ \mathcal{S}: w \mapsto \inf_{y\in \bar{y}_w+\text{Im}(\beta)^{\perp}} \mathcal{R}(y) = \inf_{u \in \text{Im}(\beta)^{\perp}} \mathcal{R}(\beta(\beta'\beta)^{-1}w+u).$$
\end{defi}

\begin{rem}
If the asset exposure risk measure is volatility, i.e. $\mathcal{R}(y) := \sqrt{y'\Sigma y}$, then $$\mathcal{S}(w) = \mathcal{R}\big(\Sigma^{-1}\beta\big(\beta'\Sigma^{-1}\beta\big)^{-1}w\big) = \sqrt{w'\big(\beta'\Sigma^{-1}\beta\big)^{-1}w}.$$
\end{rem}

\begin{rem}
$\mathcal{S}$ is convex as it can be expressed with the infimal convolution\footnote{We recall that the infimal convolution of two functions $\phi$ and $\psi$ is defined as $$\phi \square \psi (a) = \inf \{\phi(a-b) + \psi(b) | b \in \mathbb{R}^d \}.$$} of $\mathcal{R}$ and $\mathcal{X}_{\text{Im}(\beta)^{\perp}}$:
$$\mathcal{S}(w) = \mathcal{R} \square \mathcal{X}_{\text{Im}(\beta)^{\perp}}(\beta(\beta'\beta)^{-1}w),$$
where $\mathcal{X}_{F}(u)=  \begin{cases} 
      0 & u \in F \\
      +\infty & u \notin F 
   \end{cases}.$
\end{rem}

In order to study $\mathcal{S}$, we start with two propositions regarding $\mathcal{R} \square \mathcal{X}_{\text{Im}(\beta)^{\perp}}$. 

\begin{prop}
    \label{prop_unique_u}
    $\forall y\in\mathbb{R}^d$, there exists a unique $u^*(y) \in \text{Im}(\beta)^{\perp}$ such that $$\inf_{u \in \text{Im}(\beta)^{\perp}} \mathcal{R}(y+u) = \mathcal{R}(y+u^*(y)).$$
\end{prop}
\begin{proof}
    Let $m = \underset{x\in\mathbb{R}^d, \lVert x \rVert =1}\min \mathcal{R}(x)>0.$ Then, $\forall y\in\mathbb{R}^d\setminus\{\mathbf{0}\}$, $$\quad \mathcal{R}(y) =  \lVert y  \rVert \mathcal{R}\bigg(\frac{y}{\lVert y  \rVert}\bigg) \geq \lVert y  \rVert m.$$
    Therefore, $\mathcal{R}$ is coercive. 

\medskip 

We deduce that the convex function $u \in \text{Im}(\beta)^{\perp} \mapsto \mathcal{R}(y+u)$ is coercive and therefore $\forall y\in~\mathbb{R}^d$, $\exists u^* \in \text{Im}(\beta)^{\perp}$ such that $$\quad \mathcal{R}(y+u^*) = \inf_{u \in \text{Im}(\beta)^{\perp}} \mathcal{R}(y+u) = \min_{u \in \text{Im}(\beta)^{\perp}} \mathcal{R}(y+u) =: \frak m.$$

Now, if $y \in \text{Im}(\beta)^{\perp}$, $u^*$ is unique and given by $-y$ as $\mathcal R(\mathbf 0) = 0$.\\

On the contrary, if $y \notin \text{Im}(\beta)^{\perp}$, and if $u_1^*, u_2^* \in \text{Im}(\beta)^{\perp}$ are such that $$\mathcal{R}(y+u_1^*) = \mathcal{R}(y+u_2^*) = \frak m,$$ 
then,  
$$\frak m \leq \mathcal{R}\Big(y+\frac{u_1^*+u_2^*}{2}\Big) = \mathcal{R}\Big(\frac{y+u_1^*}{2}+\frac{y+u_2^*}{2}\Big) \leq \frac{\mathcal{R}(y+u_1^*)}{2}+\frac{\mathcal{R}(y+u_2^*)}{2} \le \frak m.$$

In other words, we are in the equality case of the sub-additivity property of $\mathcal R$. As $y \notin \text{Im}(\beta)^{\perp}$, $y+u_1^* \neq \mathbf{0}$ and $y+u_2^* \neq \mathbf{0}$. Therefore $$\exists c >0, \quad y+u_1^* = c (y+u_2^*).$$

Thus, $(1-c)y = c u_1^* - u_2^* \in \text{Im}(\beta)^{\perp}$ so that $c=1$ and $u_1^* = u_2^*.$
\end{proof}

\begin{prop}
The convex function $\mathcal{R} \square \mathcal{X}_{\text{Im}(\beta)^{\perp}}$ is proper, positively homogeneous on $\mathbb{R}^d$ and continuously differentiable on $\mathbb{R}^d\setminus \text{Im}(\beta)^{\perp}$ with
$$\forall y \in \mathbb{R}^d\setminus \text{Im}(\beta)^{\perp}, \quad \nabla(\mathcal{R} \square \mathcal{X}_{\text{Im}(\beta)^{\perp}})(y) = \nabla \mathcal{R}(y-u^*(y)).$$ 
\end{prop}
\begin{proof}

$\mathcal{R}$ and $\mathcal{X}_{\text{Im}(\beta)^{\perp}}$ are nonnegative and, therefore, $\mathcal{R} \square \mathcal{X}_{\text{Im}(\beta)^{\perp}} \ge 0$. Moreover, $\mathbf 0 \in \text{Im}(\beta)^{\perp}$, so $\mathcal{R} \square \mathcal{X}_{\text{Im}(\beta)^{\perp}} \le \mathcal R$. This proves that  $\mathcal{R} \square \mathcal{X}_{\text{Im}(\beta)^{\perp}}$ is proper.

\medskip

By positive homogeneity, we have  $\mathcal{R} \square \mathcal{X}_{\text{Im}(\beta)^{\perp}}(\mathbf{0})=0$, and for all $y \in \mathbb{R}^d$ and $\lambda>0$: 
\begin{align*}
    \mathcal{R} \square \mathcal{X}_{\text{Im}(\beta)^{\perp}}(\lambda y) &= \inf_{x\in\mathbb R^d} \mathcal{R}(\lambda y-x) +  \mathcal{X}_{\text{Im}(\beta)^{\perp}}(x) \\
    &=\inf_{\tilde x\in\mathbb R^d} \mathcal{R}(\lambda y-\lambda \tilde x ) +  \mathcal{X}_{\text{Im}(\beta)^{\perp}}(\lambda \tilde x ) \\
    &=\lambda \inf_{\tilde x\in\mathbb R^d} \mathcal{R}( y-\tilde x) +  \mathcal{X}_{\text{Im}(\beta)^{\perp}}(\tilde x)\\
    &=\lambda \mathcal{R} \square \mathcal{X}_{\text{Im}(\beta)^{\perp}}(y).
\end{align*}

\medskip
The dual conjugate of $\mathcal{R} \square \mathcal{X}_{\text{Im}(\beta)^{\perp}}$ is 
    $$(\mathcal{R} \square \mathcal{X}_{\text{Im}(\beta)^{\perp}})^* = \mathcal{R}^* + (\mathcal{X}_{\text{Im}(\beta)^{\perp}})^* = \mathcal{R}^* + \mathcal{X}_{\text{Im}(\beta)}.$$

For all $y,p \in \mathbb{R}^d$, we have $$p'y - \mathcal{R}(y) \leq \lVert p\rVert \lVert y\rVert - m \lVert y\rVert$$ where $m = \underset{x\in\mathbb{R}^d, \lVert x \rVert =1}\min \mathcal{R}(x)>0.$ So, $\forall p \in \bar{B}(0,m)$,  $\mathcal{R}^*(p)=0.$

\medskip

In particular, $\mathbf{0} \in \text{ri}(\text{dom}(\mathcal{R}^*))$ and $\mathbf{0} \in \text{ri}(\text{dom}\mathcal{X}_{\text{Im}(\beta)})=\text{Im}(\beta).$

\medskip

So $\text{ri}(\text{dom}(\mathcal{R}^*)) \cap \text{ri}(\text{dom}\mathcal{X}_{\text{Im}(\beta)}) \neq \emptyset$ and we have $\forall p \in \text{dom}(\mathcal{R}^*) \cap \text{Im}(\beta)$ that\footnote{$\partial f(x)$ denotes the subdifferential, the set of all subgradients of $f$ at $x$.} $$ \partial(\mathcal{R}^*+\mathcal{X}_{\text{Im}(\beta)})(p)=\partial \mathcal{R}^*(p) +\partial \mathcal{X}_{\text{Im}(\beta)}(p) = \partial \mathcal{R}^*(p) + \text{Im}(\beta)^{\perp}.$$

    Now, $\forall y \in \mathbb{R}^d\setminus \text{Im}(\beta)^{\perp}$,
    \begin{align*}
        p \in \partial (\mathcal{R} \square \mathcal{X}_{\text{Im}(\beta)^{\perp}})(y)&\iff y \in  \partial ((\mathcal{R} \square \mathcal{X}_{\text{Im}(\beta)^{\perp}})^*)(p) \\
        &\iff 
         \begin{cases} 
              p \in \text{dom}(\mathcal{R}^*) \cap \text{Im}(\beta) \\
              y \in \partial \mathcal{R}^*(p) + \text{Im}(\beta)^{\perp} 
   \end{cases} \\
        &\iff 
         \begin{cases} 
              p \in \text{dom}(\mathcal{R}^*) \cap \text{Im}(\beta) \\
              \exists u \in \text{Im}(\beta)^{\perp}, y-u \in \partial \mathcal{R}^*(p) 
   \end{cases} \\
        &\iff 
         \begin{cases} 
              p \in \text{dom}(\mathcal{R}^*) \cap \text{Im}(\beta) \\
              \exists u \in \text{Im}(\beta)^{\perp}, p = \nabla \mathcal{R}(y-u) 
   \end{cases}
    \end{align*}
    as $y-u \neq \mathbf{0}$, since $y \notin \text{Im}(\beta)^{\perp}.$

\medskip
Therefore, if $p \in \partial (\mathcal{R} \square \mathcal{X}_{\text{Im}(\beta)^{\perp}})(y)$, then  $\exists u \in \text{Im}(\beta)^{\perp}$ such that $p = \nabla \mathcal{R}(y-u) \in \text{Im}(\beta).$ This is exactly the first order condition of the convex optimization problem of Proposition \ref{prop_unique_u}. Therefore, $u=u^*(y)$ and $p=\nabla \mathcal{R}(y-u^*(y))$. Differentiability then comes from the fact that subdifferentials are singletons on $\mathbb{R}^d\setminus \text{Im}(\beta)^{\perp}$ and continuous differentiability is a consequence of convexity. In particular $$\nabla(\mathcal{R} \square \mathcal{X}_{\text{Im}(\beta)^{\perp}})(y) = \nabla \mathcal{R}(y-u^*(y)).$$ 
\end{proof}

From the above Propositions, we deduce the following Proposition on $\mathcal{S}$ which is central for our analysis.\\

\begin{prop}
\label{corolarry1}
$\mathcal{S}$ is positively homogeneous, convex on $\mathbb{R}^m$ and continuously differentiable on $\mathbb{R}^m\setminus\{\mathbf{0}\}$. Moreover, 
$$\forall w \neq \mathbf{0}, \quad \nabla \mathcal{S}(w) = (\beta'\beta)^{-1} \beta' \nabla \mathcal{R}\big(\beta(\beta'\beta)^{-1}w + u^*(\beta(\beta'\beta)^{-1}w)\big).$$
\end{prop}
\begin{proof}
This result is a consequence of the definition of $\mathcal S$ and the above propositions since, for $w \neq \mathbf 0$, $\beta(\beta'\beta)^{-1}w \in \text{Im}(\beta) \setminus \{\mathbf 0\} \subset \mathbb{R}^d\setminus \text{Im}(\beta)^{\perp}$.\\
\end{proof}

The function $\mathcal S$ allows to compute the minimum risk achievable for a given vector of factor exposures, and thus assesses the suitability of a given factor allocation scheme, for example, in terms of risk compliance. By Proposition \ref{prop_unique_u}, this minimum risk is associated with a unique vector of asset exposures given by $y^*(w) = \beta(\beta'\beta)^{-1}w + u^*(\beta(\beta'\beta)^{-1}w)$.\\

The fact that $\mathcal S$ 
is positively homogeneous and continuously differentiable on $\mathbb{R}^m\setminus\{\mathbf{0}\}$ enables to compute the sensitivity of the portfolio risk to the factors for a given vector of factor or asset exposures.

\medskip
Given a vector $w$ of factor exposures, Euler homogeneous function theorem leads to a breakdown of risk as\footnote{
As stated in Proposition \ref{corolarry1}, the risk contribution of the factor $i$, i.e. $w_i \partial_i \mathcal{S}(w)$, can be easily calculated using $\nabla \mathcal{R}$. When the risk measure is volatility, this gradient has an analytic form.  As for the sensitivity of Expected Shortfall, general expressions exist (see \cite{gourieroux2000sensitivity} and \cite{tasche1999risk}).}
$$\mathcal{S}(w) = \sum_{i=1}^m w_i \partial_i \mathcal{S}(w) = w' \nabla \mathcal{S}(w) .$$
This measure of risk $\mathcal{S}(w)$ corresponds to $\mathcal{R}(y^*(w))$: the risk associated with the optimal (with respect to $\mathcal R$) portfolio with factor exposures $w$. Now, given a vector of asset exposures $y$, we can consider the associated factor exposures $\beta'y$ and examine the breakdown of risk:
\begin{eqnarray*}
\mathcal{R}(y) &=& \mathcal{R}(y) - \mathcal{R}(y^*(\beta'y)) + \mathcal{R}(y^*(\beta'y))\\
&=& \underbrace{\mathcal{R}(y) - \mathcal{R}(y^*(\beta'y))}_{\text{(I)}} + \underbrace{\mathcal{S}(\beta'y)}_{\text{(II)}}.\\
\end{eqnarray*}
(I) is a nonnegative term, equal to $0$ if and only if $y = y^*(\beta'y)$. It corresponds to the risk of not choosing the optimal portfolio given factor exposures. (II) corresponds to the risk associated with the risk factors as measured by $\mathcal S$. Of course, (II) can be broken down as above across risk factors using Euler homogeneous function theorem.

\medskip
In addition to evaluating portfolio risk through a factor-based approach, we can take it a step further and actively manage the risk contributions of the factors using the concept of Risk Budgeting introduced in Section~\ref{intro_RB}, though with $\mathcal S$ replacing $\mathcal R$ because they share the same fundamental properties needed for Risk Budgeting (see Proposition \ref{corolarry1}). This is the purpose of the next section.

\section{Factor Risk Budgeting: managing risk contributions of factors}
\label{section4}
\subsection{Framing the problem of Factor Risk Budgeting}
In the previous section, we established that it is possible to introduce a factor risk measure, i.e. a function mapping a vector of factor exposures to a risk figure. We also demonstrated that this function adheres to all the necessary properties of asset exposure risk measures, particularly in the context of Risk Budgeting, making it well-suited for portfolio construction.

\medskip
By analogy with Risk Budgeting, we can use $\mathcal S$ to define a Factor Risk Budgeting portfolio allocating risk across the underlying risk factors given factor risk budgets. Mathematically, we define the Factor Risk Budgeting (FRB) problem as follows: 

\begin{defi}
\label{FRB_problem_def}
Let $\beta \in \mathbb R^{d \times m}$ be a full rank matrix of factor loadings. Let $b \in \Delta^{>0}_m$ be a vector of factor risk budgets. We say that a vector of weights $\theta \in \Delta_d$, solves the Factor Risk Budgeting problem $\text{FRB}_{b}$ if and only if 
$$\quad w_i \partial_i \mathcal{S}(w) = b_i \mathcal{S}(w),$$ 
for every $i \in \{1,\ldots,m\}$ where $w = \beta'\theta \in (\Rb_+^*)^m$.
\end{defi}

\begin{rem}
\label{rem_FRB}
In the above definition, the vector $w$ of factor exposures has positive coordinates. However, Factor Risk Budgeting portfolios are not necessarily long-only portfolios. 
\end{rem}

By analogy with Theorem \ref{thm_existence}, we can state the following result:

\begin{thm}
\label{thm_existence_FRM}
Let $b \in \Delta^{>0}_m$.

Let $g : \mathbb{R_+} \rightarrow \mathbb{R}$ be a continuously differentiable convex and increasing function.
Let the function $\Gamma_g:(\Rb_+^*)^m \rightarrow \Rb$ be defined by
\vspace{-3mm} 
$$ \Gamma_g: w \mapsto g\big(\mathcal{S}(w)\big) - \sum_{i=1}^m b_i \log{w_i}.$$
There exists a unique minimizer $w^*$ of the function $\Gamma_g$.

Then, $\forall y \in y^*(w^*) + \text{Im}(\beta)^\perp$, $\sum_{i=1}^d y_i \neq 0 \implies \theta(y) = \frac{y}{\sum_{i=1}^d y_i} \in \Delta_d$ is solution of~$\text{FRB}_{b}$.\\
Moreover, if $\theta \in \Delta_d$ is solution of $\text{FRB}_{b}$, then $w = \beta'\theta$ is proportional to $w^*$.
\end{thm}

From Theorem \ref{thm_existence_FRM}, solving $\text{FRB}_{b}$ boils down to solving a  minimization problem of the form
$$\min_{w \in (\mathbb R_+^*)^m} g\big(\mathcal{S}(w)\big) - \sum_{i=1}^m b_i \log{w_i}.$$
However, as $g$ is increasing, we have 
\begin{align}
\label{frb_problem_general}
\min_{w \in (\mathbb R_+^*)^m} g\big(\mathcal{S}(w)\big) - \sum_{i=1}^m b_i \log{w_i} &=\min_{w \in (\mathbb R_+^*)^m} g\Bigg(\min_{y\in\mathbb{R}^d, \beta'y=w} \mathcal{R}(y)\Bigg) - \sum_{i=1}^m b_i \log{w_i} \nonumber \\ 
&= \min_{w \in (\mathbb R_+^*)^m} \min_{y\in\mathbb{R}^d, \beta'y=w} g(\mathcal{R}(y)) - \sum_{i=1}^m b_i \log{w_i}\nonumber \\
&=\min_{y \in \mathcal C} g(\mathcal{R}(y)) - \sum_{i=1}^m b_i \log{(\beta'y)_i}
\end{align}
where $\mathcal C = \big\{y\in\mathbb{R}^d | (\beta'y)_i>0, \forall i \in \{1,\ldots,m\}\big\}$.\\

In other words, the fact that $\mathcal S$ is itself defined as a minimum depending on $\mathcal R$ allows to solve $\text{FRB}_{b}$ through a minimization problem involving $\mathcal R$ only (not $\mathcal S$). Moreover, solving this minimization problem leads to a vector of asset exposures that is the best in terms of risk, given the associated factor exposures.

\medskip

Now, if $g\big(\mathcal{R}(y)\big) = \min_{\zeta \in \mathcal{Z}} \mathbb E\big[H(\zeta, -y'X) \big]$ for some set $\mathcal{Z}$ and some function $H$, the above problem writes
\begin{equation}
\label{frb_stochastic}
\min_{y \in \mathcal C} \min_{\zeta \in \mathcal{Z}} \mathbb E\big[H(\zeta, -y'X) \big] - \sum_{i=1}^m b_i \log{(\beta'y)_i}=  \min_{y \in \mathcal C, \zeta \in \mathcal Z} \mathbb E\Big[H(\zeta, -y'X) - \sum_{i=1}^m b_i \log{(\beta'y)_i}\Big].  
\end{equation}

If $(y^*,\zeta^*)$ is a solution of the above stochastic optimization problem, then, as soon as $\sum_{i=1}^d y^*_i \neq 0$,  $\theta := y^*/\sum_{i=1}^d y^*_i$ solves $\text{FRB}_b$.

\medskip

We clearly see that the computational methodologies that are applicable to the classical Risk Budgeting problem are equally suitable for the computation of Factor Risk Budgeting portfolios. Moreover, the same set of risk measures as for Risk Budgeting can be used for Factor Risk Budgeting. 

\subsection{The case of long-only portfolios}
\label{long_only_case}
It is essential to note that Factor Risk Budgeting portfolios presented above constitute fully invested general portfolios. In practice, short positions are not always permitted, and asset weights generally have to comply with certain constraints. Especially, it is often the case that portfolios are restricted to be long only. We therefore highlight two possible approaches for the case of long-only portfolios in an FRB setting.

\medskip
Solving Problem (\ref{frb_problem_general}) or (\ref{frb_stochastic}) results in an FRB portfolio $\theta \in \Delta_d$ and a vector of factor exposures~$w^*$ that satisfy the condition set out in the Definition~\ref{FRB_problem_def}. Since $\theta$ is not necessarily long only, one possible approach is to look for an alternative portfolio that respects the constraints and produces the same factor exposures $w^*$. Precisely, this can be done on the basis of our initial risk-based approach by solving  
$$\underset{y\in\mathbb{R}_+^d, \beta'y=w^*}\min \mathcal{R}(y)$$
and normalizing its solution. Such an approach allows us to search for an FRB portfolio satisfying our constraint. However, it requires the set $\big\{y\in\mathbb{R}_+^d | \beta'y=w^*\big\}$ to be non-empty, which is not necessarily true in all cases. 

\medskip

Another approach can therefore be to incorporate the long-only constraint directly in Problem~(\ref{frb_problem_general}) or (\ref{frb_stochastic}) and consider the set $\mathcal C^{\geq 0} = \big\{y\in\mathbb{R}_+^d | (\beta'y)_i>0, \forall i \in \{1,\ldots,m\}\big\}$ for the variable~$y$.\footnote{Even the set $\mathcal C^{\geq 0}$ is not always non-empty, meaning that there may not even exist a nonnegative vector of asset exposures leading to a positive vector of factor exposures. A simple approach to establish this by contradiction is to consider $\beta$ whose elements are negative. In such a scenario, no nonnegative vector $y$ can lead to a positive vector $w = \beta'y$. However, this constraint is more of a theoretical limitation than a practical one: in most practical situations, $\mathcal C^{\geq 0}$ is a non-empty set.} There, of course, still exists a unique minimizer that solves the aforementioned problem. Nevertheless, $\theta := y^*/\sum_{i=1}^d y^*_i \in \Delta^{\geq 0}_d$, now, may not guarantee a solution to $\text{FRB}_b$. On the other hand, we expect this to be a portfolio in which the factor risk contributions are close to the factor risk budgets. We consider this latter approach in the numerical section of this paper and analyse the proximity of factor risk contributions to factor risk budgets for long-only FRB portfolios.

\section{Balancing asset and factor risk contributions}
\label{asset_factor_setting}

\subsection{Asset-Factor Risk Budgeting portfolios}
\label{asset_factor_theory}
Although FRB portfolios may present a theoretical appeal, they can run into a number of practical pitfalls. The search for a precise match in factor risk budgets can lead to concentrated portfolios in asset weights, especially in the original long-short setting, amplifying individual asset risk contributions. Likewise, classical RB portfolios may contain hidden risks associated with the underlying risk factors since the sole focus is to satisfy desired risk contribution at the asset level. 

\medskip

Consequently, a framework benefiting from both ideas can be of practical use to construct portfolios which take asset and factor risk contributions into account together in the same setting. The fact that FRB portfolios can be found by solving an optimization problem (Problem (\ref{frb_problem_general}) or (\ref{frb_stochastic})) involving only $\mathcal{R}$ allows us to propose an optimization framework that seeks a balanced solution between FRB and classical RB portfolios. 

\medskip

More precisely, for given vectors representing asset risk budgets $b_a \in \Delta^{>0}_d$ and factor risk budgets $b_f \in \Delta^{>0}_m$, we can consider  the following optimization problem: 
\begin{equation}
\label{asset_factor_rb}
\min_{y \in \mathcal C^{> 0}} g(\mathcal{R}(y))  -  \lambda_a \sum_{i=1}^d {b_a}_i \log{y_i} - \lambda_f \sum_{i=1}^m {b_f}_i \log{(\beta'y)_i},     
\end{equation}
where $\mathcal C^{>0} = \big\{y\in(\mathbb R_+^*)^d | (\beta'y)_i>0, \forall i \in \{1,\ldots,m\}\big\}$ and $\lambda_a, \lambda_f \in \mathbb R_+^*$ are asset and factor importance parameters. 

\medskip

Of course, if $g\big(\mathcal{R}(y)\big) = \min_{\zeta \in \mathcal{Z}} \mathbb E\big[H(\zeta, -y'X) \big]$ for some set $\mathcal{Z}$ and some function $H$, the above problem writes
\begin{equation}
\label{asset_factor_rb_stochastic}
 \min_{y \in \mathcal C^{> 0}, \zeta \in \mathcal Z} \mathbb E\Big[H(\zeta, -y'X) -  \lambda_a \sum_{i=1}^d {b_a}_i \log{y_i} - \lambda_f \sum_{i=1}^m {b_f}_i \log{(\beta'y)_i}\Big].  
\end{equation}

If $(y^*,\zeta^*)$ is a solution of the above stochastic optimization problem, then, $\theta := y^*/\sum_{i=1}^d y^*_i$ is called the Asset-Factor Risk Budgeting (AFRB) portfolio.

\medskip
As a result, the AFRB portfolio is obtained by solving an optimization problem that involves minimizing a risk term subject to two logarithmic barriers. The risk measure and the barriers naturally prevent asset and factor exposures from diverging and approaching zero, ensuring that the resulting portfolios are balanced in terms of both exposure types and associated risks. By adjusting the importance parameters, one can introduce varying degrees of sensitivity to different types of risk contributions, thus enabling the introduction of tilts towards certain aspects.

\medskip
We should note that AFRB portfolios never exactly match the given risk budgets for assets and factors. However, 
this formulation allows us to explore a compromise solution that reconciles the advantages of FRB portfolios as well as penalizing large deviations from asset risk budgets.  

\subsection{Illustration with a simple model}
\label{asset_factor_toy_model}
Let us examine a simple model with 4 assets and 3 factors to gain insight into the proposed Asset-Factor Risk Budgeting framework and its distinctions from FRB and classical RB. For simplicity, we adopt the volatility risk measure and assume the following arbitrary parameters.\footnote{The parameters are taken from the example presented in \cite{roncalli2016risk}.}

\begin{center}
$\beta\! =\! \begin{bmatrix}
0.9 &  0.0 &  0.5 \\
1.1 &  0.5 &  0.0 \\
1.2 &  0.3 &  0.2 \\
0.8 &  0.1 &  0.7 \\
\end{bmatrix}$ and
$\Sigma\! =\! \begin{bmatrix}
0.0449 &  0.0396 &  0.0442 &  0.0323 \\
0.0396 &  0.0734 &  0.0543 &  0.0357 \\
0.0442 &  0.0543 &  0.0689 &  0.0401 \\
0.0323 &  0.0357 &  0.0401 &  0.0531 \\
\end{bmatrix}$.
\end{center}

With these specifications, we can compute the RB, FRB and AFRB portfolios through Problems~(\ref{rb_problem_vol}), (\ref{frb_problem_general}) and (\ref{asset_factor_rb}), respectively, based on the risk budgets $b_a =(0.25,0.25,0.25,0.25)$ for assets and $b_f = (1/3, 1/3, 1/3)$ for factors. For AFRB, we also require the importance parameters $(\lambda_a, \lambda_f)$, which we set to $(0.2, 0.8)$ in this simple example.

\medskip
Table~\ref{asset_based} illustrates the portfolios computed under the aforementioned setting. As expected, the RB portfolio evenly distributes its risk across the assets. In contrast, the FRB portfolio appears concentrated on two assets, both in terms of weights and risk contributions. In particular, it has the highest total risk across our three portfolios. We clearly see that the AFRB portfolio prevents excessive asset risk concentration within the portfolio.

\begin{center}
\begin{tabular}{c|ccc|ccc}
\toprule
 & \multicolumn{3}{c|}{$\theta_i$} & \multicolumn{3}{c}{$\theta_i \partial_i \mathcal{R}(\theta)$} \\
{Asset} &           RB &     FRB &    AFRB &                    RB &    FRB &   AFRB \\
\midrule
 1    &        27.86 &   -6.60 &   18.26 &                  5.28 &  -1.05 &   3.33 \\
 2    &        22.60 &   34.95 &   25.72 &                  5.28 &   7.93 &   6.00 \\
 3    &        21.98 &    8.87 &   17.97 &                  5.28 &   1.90 &   4.22 \\
 4    &        27.56 &   62.78 &   38.05 &                  5.28 &  13.38 &   7.63 \\
\midrule
$\sum_{1 \le i \le 4}$  &       100.00 &  100.00 &  100.00 &                 21.13 &  22.16 &  21.18 \\
\bottomrule
\end{tabular}
\captionof{table}{Asset weights and associated asset risk contributions (in \%) of the portfolios computed using the different methods.}
\label{asset_based}
\end{center}

To conduct a similar analysis from the perspective of the underlying risk factors, we refer to Table~\ref{factor_based}. Firstly, for the three portfolios, we observe significant differences in risk allocation across factors, although they have similar factor exposures. Factor risk contributions in the FRB portfolio are perfectly in line with the chosen risk budgets, potentially leading to the previously observed unbalanced asset weights. Conversely, the RB portfolio, lacking sensitivity to factors by definition, is highly exposed to the risk of the first factor. The AFRB portfolio, designed with sensitivity to both assets and factors, mitigates that effect with only a negligible increase in the total portfolio risk.

\begin{center}
\begin{tabular}{c|ccc|ccc}
\toprule
{} & \multicolumn{3}{c|}{$w_i$} & \multicolumn{3}{c}{$w_i \partial_i \mathcal{S}(w)$} \\
{Factor} &           RB &     FRB &    AFRB &                    RB &    FRB &   AFRB \\
\midrule
 1   &        98.36 &   93.37 &   96.73 &                 16.64 &   7.39 &  13.87 \\
 2   &        20.65 &   26.42 &   22.06 &                  1.80 &   7.39 &   3.22 \\
 3   &        37.62 &   42.42 &   39.36 &                  2.66 &   7.39 &   4.08 \\
\midrule
$\sum_{1 \le i \le 3}$ &       156.63 &  162.21 &  158.15 &                 21.11 &  22.16 &  21.17 \\
\bottomrule
\end{tabular}
\captionof{table}{Factor exposures and associated factor risk contributions (in \%) of the portfolios computed using the different methods.}
\label{factor_based}
\end{center}

\subsection{Choosing the importance parameters}

The selection of the importance parameters is a critical aspect. Ultimately, Asset-Factor Risk Budgeting aims to identify portfolios whose normalized risk contributions closely align with a reference vector of risk budgets. Hence, it is essential to assess how far we deviate from the risk budgets in terms of risk contributions for given values of $(\lambda_a, \lambda_f)$.

\medskip

Relative entropy serves as a useful tool for conducting such analyses since risk budgets are strictly positive and sum up to 1, making them a suitable reference distribution for risk contributions.\footnote{We assume nonnegative risk contributions, which is typically the case for long-only portfolios in assets and factors. In the event of negative risk contributions, one can utilize a norm as a distance measure to conduct a similar analysis.} Specifically, we can compute an asset-based relative entropy score $\sum_{i=1}^d q_i \log{(q_i/b_{a_i})}$ and a factor-based relative entropy score $\sum_{i=1}^m s_i \log{(s_i/b_{f_i})}$, respectively, where $q_i = \theta_i\partial_i\mathcal{R}(\theta)/\mathcal{R}(\theta)$ and $s_i = w_i\partial_i\mathcal{S}(w)/\mathcal{S}(w)$. These scores provide insights into the degree of deviation from the target risk contributions and can be helpful to determine the importance parameters. By selecting a pair where both relative entropy scores are low, one can identify the importance parameters that result in a small amount of deviation from the desired risk allocations.

\begin{figure}[!ht]
\centering
     \includegraphics[width=1\textwidth]{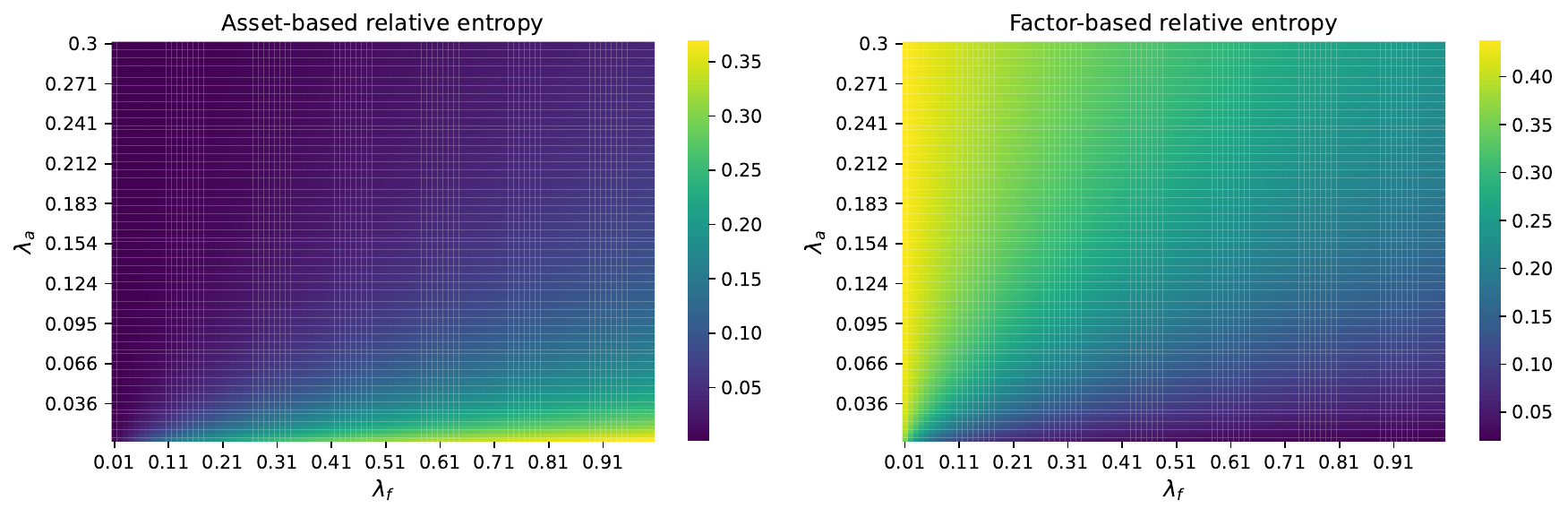}
      \caption{The impact of the choice of $(\lambda_a,\lambda_f)$ on the distance of normalized risk contributions from specified risk budgets at the asset level (left) and factor level (right).}
      \label{chart_diver}
\end{figure}
\vspace{0mm}

Figure~\ref{chart_diver} illustrates the relative entropy scores of risk contributions for the AFRB portfolios computed across a range of different parameter values in the above model. As expected, as $\lambda_a$ increases, normalized asset risk contributions tend to converge towards $b_a$, evidenced by the asset-based relative entropy score approaching zero. However, this convergence is associated with an increase in factor-based relative entropy, indicating a shift in factor risk contributions. Conversely, to achieve normalized factor risk contributions that closely match $b_f$, one should choose $\lambda_f$ to be greater than $\lambda_a$. Therefore, it is crucial to ensure that the selected parameters maintain a good balance between both relative entropy scores. Additionally, one can adopt a systematic approach to conduct a grid search for a point that minimizes a combined score, e.g. the mean of asset-based and factor-based relative entropy scores.

\section{Numerical applications with financial data}
\label{section6}
\subsection{Analysis of risk-based portfolios in an equity factor model}
In this section, we construct various risk-based portfolios using the BARRA US Equity (USE3) model.\footnote{A detailed description of the model can be found online or requested through \url{https://www.msci.com/www/research-report/united-states-equity-model-use3/015920189}.} The USE3 is a multi-factor model covering 67 factors: 54 industry factors and 13~style factors (e.g. momentum, growth and size). The covariance matrix of the factor returns is provided, revealing the estimated relationship between the factors. The model covers a wide range of US equities, and provides the factor loading of each stock on each risk factor, as well as the idiosyncratic components for individual assets.\footnote{It is worth noting that they implement a multi-industry approach, which means that stocks can have loadings on multiple industry factors. This approach enhances the model.} It therefore corresponds to having a factor model of the following form: 
$$\Sigma = \beta \Sigma_{f} \beta' + \Sigma_{idio},$$
where $\Sigma_f$ is the covariance matrix of factor returns and $\Sigma_{idio}$ is a diagonal matrix containing idiosyncratic variances of asset returns. 

\medskip

These parameters are sufficient to compute portfolios using risk-based methods when the risk measure is chosen to be volatility. It is indeed interesting to use the covariance matrix from a specialized proprietary model to construct portfolios for any risk-based method, but it is particularly interesting in the case of FRB and its variants since we also have the matrix of factor loadings.   

\medskip

We consider the stocks in the S\&P500 index as of 31/10/2023 as our investment universe.\footnote{The index actually has 503 components, as three of them have two share classes listed. We filter them and end up with a universe of 500 assets.} We also use the parameters of the BARRA model for the same date. We calculate two Minimum Variance portfolios, one without constraint and another one with a long-only constraint, denoted by $\text{MV}$ and $\text{MV}^+$ respectively. We also compute an RB portfolio (solving Problem~(\ref{rb_problem_vol})) where the risk budgets for each asset are equal -- called the $\text{ERC}$ portfolio. To introduce factor risk sensitivity into the portfolio construction process, we consider three different portfolios. The first is an FRB portfolio, in which factor risk budgets are chosen to be equal and computed by solving Problem~(\ref{frb_problem_general}). We call it the Equal Factor Risk Contribution portfolio, $\text{EFRC}$. The second portfolio is similar, but we restrict the weights to be positive following the ideas presented in Section~\ref{long_only_case}. In other words, it is the portfolio obtained by solving Problem~(\ref{frb_problem_general}) under the constraint associated with the set $\mathcal C^{\geq 0}$. We denote it by $\text{EFRC}^+$. We also consider the Equal Asset-Factor Risk Contribution portfolio, $\text{EAFRC}$ which aims to find a balance between factor risk contributions and asset risk contributions as formalized by Problem~(\ref{asset_factor_rb}) where we use equal asset and factor risk budgets and the following importance parameters: $\lambda_a = 0.3$ and $\lambda_f = 0.7$. Finally, we also include the equally-weighted portfolio, $\text{EW}$, in our analysis which is often a competitive challenger of risk-based methods.        

\begin{figure}[!ht]
\centering
     \includegraphics[width=.9\textwidth]{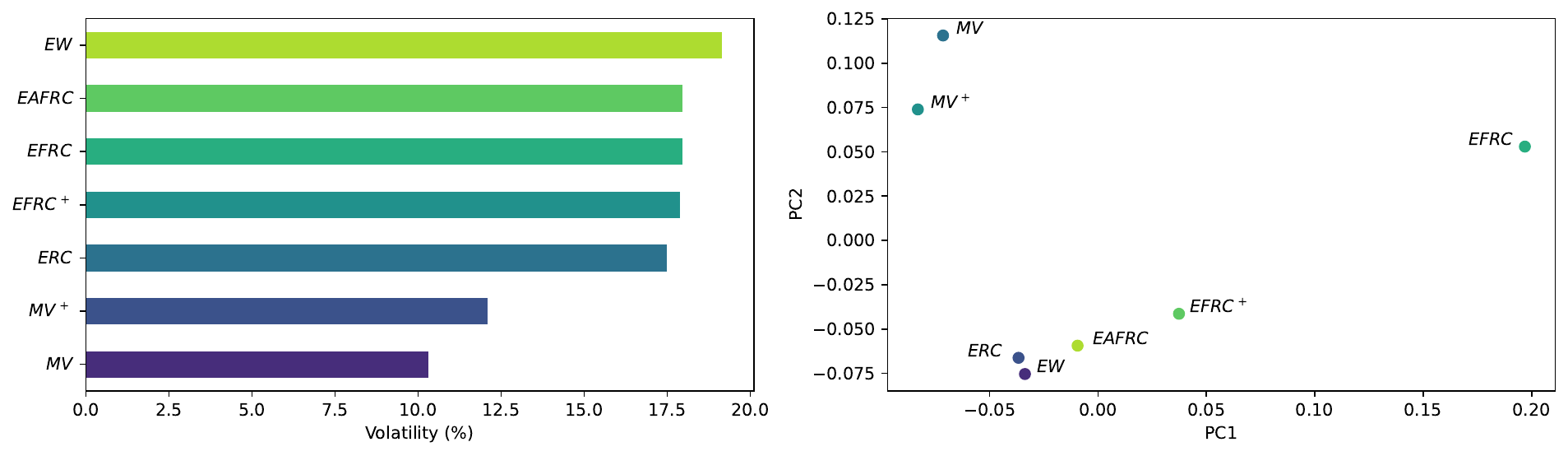}
      \caption{Annualized ex-ante volatilities (left) and latent space representations (right) of the portfolios.}
       \label{chart1}
\end{figure}
\vspace{0mm}

Figure~\ref{chart1} gives an initial idea of how the different portfolios produced by the above methods compare. We observe that the portfolios that are based on risk contributions have similar ex-ante volatilities, while they appear very different in terms of positioning in the graph on the right obtained by applying a PCA on the asset weights to locate the portfolios in a two-dimensional space. $\text{EFRC}$ seems to deviate from the other portfolios in order to perfectly match its factor risk budgets. However, adding long-only constraint ($\text{EFRC}^+$) and introducing asset risk budgets into the picture ($\text{EAFRC}$) shifts the resulting portfolios towards $\text{ERC}$ and $\text{EW}$ in terms of positioning. Unsurprisingly, Minimum Variance portfolios, $\text{MV}$ and $\text{MV}^+$ present the lowest ex-ante risks and form a cluster far away from the other portfolios in the latent space. 

\medskip
Furthermore, to better understand portfolios from the point of view of diversification both in terms of weights and risk contributions, we can plot the Lorenz curves which represent the cumulative sum of the sorted elements of the different vectors. 

\medskip

In Figure~\ref{LorenzAsset}, we plot Lorenz curves based on the vector of asset weights and the vector of asset risk contributions for portfolios of interest. In the graph on the left, the curve for $\text{EW}$ lies on the \textit{equality} line and can be regarded as the \textit{least-concentrated} portfolio in terms of asset weights, since the capital is evenly distributed across assets. We observe that $\text{ERC}$ and $\text{EAFRC}$ portfolios are not very far from this line, indicating that the asset weights do not significantly depart from each other for these portfolios. On the contrary, $\text{MV}^+$ and $\text{EFRC}^+$ produce portfolios where the majority of assets have weights close to zero, resulting in  more concentrated portfolios. Assessing diversification for long-short portfolios is not straightforward, however, and we observe that the proportions of long and short positions are similar for the $\text{MV}$ and $\text{EFRC}$ portfolios, resulting in $U$-shaped Lorenz curves. In the graph on the right, we see a very similar picture, except that $\text{ERC}$ is obviously now the least-concentrated portfolio in terms of asset risk contributions, spreading portfolio risk evenly across assets. $\text{EAFRC}$ is close to both  $\text{EW}$ and  $\text{ERC}$ on both graphs. 

\begin{figure}[!ht]
\centering
     \includegraphics[width=\textwidth]{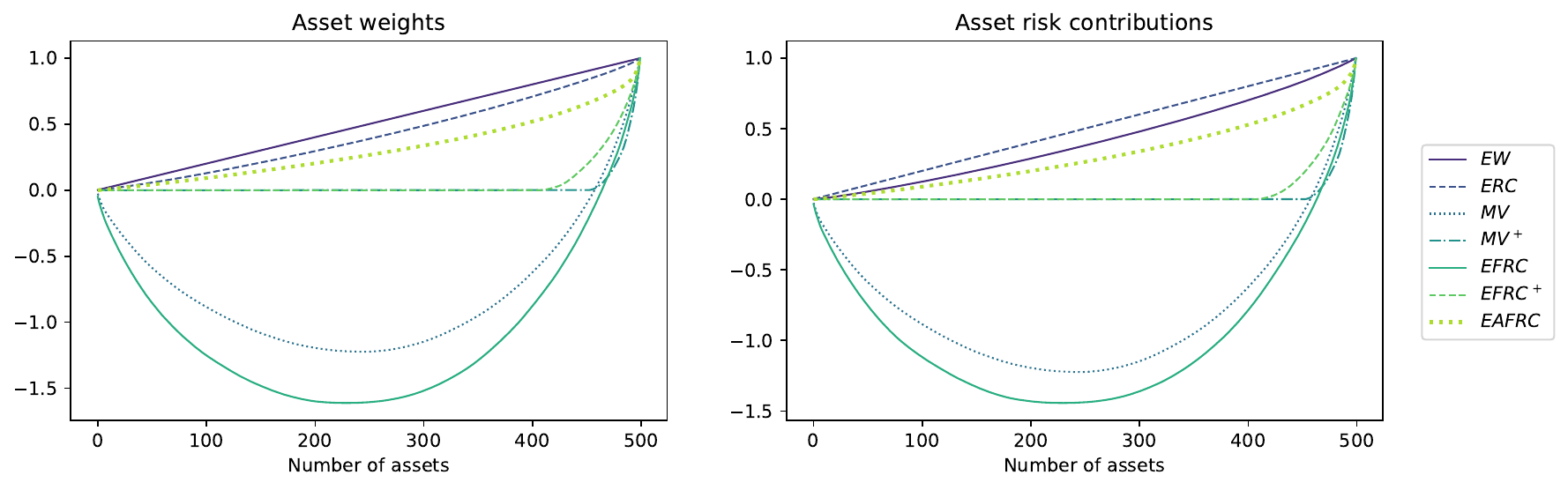}
      \caption{Lorenz curves of asset weights (left) and asset risk contributions (right) (in \%) for different portfolios}
       \label{LorenzAsset}
\end{figure}
\vspace{0mm}

The results of the same analysis on factor risk contributions (using Proposition~\ref{corolarry1}) are depicted in Figure~\ref{LorenzFactor}. This time, $\text{EFRC}$ plays the benchmark role as the least-concentrated portfolio in terms of factor risk contributions, distributing risk across factors equally. $\text{EFRC}^+$ and $\text{EAFRC}$ follow this portfolio, providing factor risk distributions that are close to factor risk budgets. $\text{EW}$, $\text{ERC}$, $\text{MV}^+$ and $\text{MV}$ are portfolios that do not allow factor risk contributions to be controlled by definition, and can result in portfolios where risk is concentrated on a small set of factors.

\begin{figure}[!ht]
\centering
     \includegraphics[width=.6\textwidth]{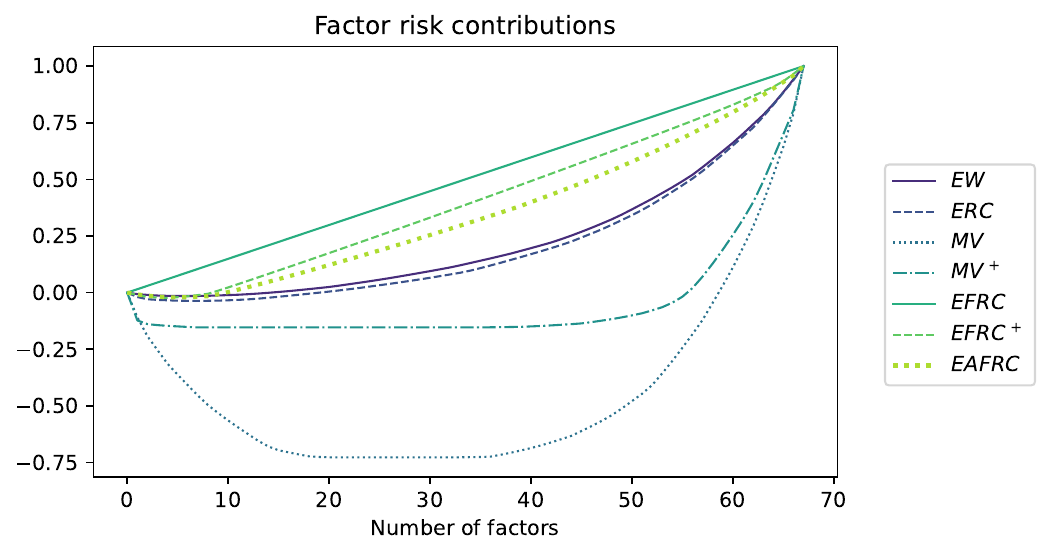}
      \caption{Lorenz curves of factor risk contributions (in \%) for different portfolios}
       \label{LorenzFactor}
\end{figure}
\vspace{0mm} 

\medskip

In short, we see that the $\text{EW}$ and $\text{ERC}$ portfolios are well-balanced in terms of asset weights and asset risk contributions, while the Lorenz curves for factor risk contributions indicate a potential concentration of factor risk. The $\text{EFRC}$ portfolio achieves perfectly an equal allocation of factor risk, but may suffer from a concentration of asset weights and asset risk contributions, as does the $\text{EFRC}^+$ portfolio. However, the $\text{EAFRC}$ portfolio is, as expected, interesting for balancing asset and factor risk contributions. It seems to be the only strategy in our case that achieves a balance across all three attributes. Consequently, Asset-Factor Risk Budgeting portfolios may be strong candidates for managing the relationship between the three different pillars of portfolio risk management. 

\subsection{Constructing diversified cross-asset portfolios}
In this section, we illustrate a case where we construct cross-asset portfolios within the framework of a macroeconomic factor model. Different asset classes can be affected by common macroeconomic factors, and it can therefore be interesting to manage the risk exposure to these factors when managing a cross-asset portfolio. 

\medskip
This often requires portfolio managers to first define the set of factors of interest to their investment universe. Then, as most macroeconomic factors are only observable at very low frequency (quarterly or annually), they should find proxy portfolios that approximate the returns / movements of these factors, so as to be able to perform a sensitivity analysis and find the factor loadings of the assets. They can then optimize portfolios or perform a risk analysis for a given portfolio based on a chosen risk measure. 

\medskip

We consider an investment universe of 15 assets\footnote{The corresponding tickers are respectively TY1 Comdty, CN1 Comdty, RX1 Comdty, G 1 Comdty,	ES1 Index,	RTY1 Index,	VG1 Index,	SXE1 Index, Z 1 Index, NH1 Index, MES1 Index, CRUD LN Equity,	IGLN LN Equity and COPA LN Equity.} corresponding to 4 government bond indices (U.S. 10-Year Government Bond Index, Canada 10-Year Government Bond Index, Germany 10-Year Government Bond Index, U.K. 10-Year Government Bond Index), 7 equity indices (U.S. Large-Cap Equity Index, U.S. Small-Cap Equity Index, Eurozone Large-Cap Equity Index, Eurozone Small-Cap Equity Index, U.K. Equity Index, Japanese Equity Index, Emerging Markets Equity Index), and 3 commodities (Crude Oil, Gold, Copper). We focus on 4 macroeconomic factors of interest called Growth, Real Rates, Broad Commodities and Emerging Markets. These factors are respectively represented by the S\&P Developed BMI Index, the S\&P Global Developed Sovereign Inflation-Linked Bond Index, the S\&P GSCI Index, and an index long emerging markets equity (S\&P Emerging BMI Index), short developed markets equity (S\&P Developed BMI Index) and long emerging market currencies (MSCI Intl. Emerging Market Currency). For the latter, equity and FX indices are used in the same factor to better capture the overall risk of emerging markets.\footnote{The S\&P indices can be found on S\&P Dow Jones Indices public website: \url{https://www.spglobal.com/spdji/en/}.} The choice of risk factors and their proxies follows the model in \cite{greenberg2016factors}.

\medskip
Our objective is to investigate three strategies, namely Risk Budgeting, long-only Factor Risk Budgeting and Asset-Factor Risk Budgeting, where Expected Shortfall is chosen as the risk measure. More specifically, we want to analyse the performance of these strategies, but especially the evolution of asset and factor risk contributions over time by performing a backtest where we rebalance our portfolio on a weekly basis to obtain the desired portfolios based on the new information. 

\medskip
Our database consists of daily returns between January 2015 and November 2023. We consider a 5-year look-back window, which implies that at each rebalancing, we take into account data from the last 5 years. Consequently, the first portfolio is calculated in January 2020. At each rebalancing, we run an OLS regression using the weekly returns of each asset and the factors to obtain the matrix of factor loadings.\footnote{If the associated relationship between assets and factors is not statistically significant ($p$-value greater than~$0.05$), we consider the associated factor loading to be equal to $0$ for the given period.}  

\medskip
First, we look at the case of the RB strategy. The risk budgets are chosen as $$b_a=\left(0.1,0.1,0.1,0.1,\frac{0.4}{7},\frac{0.4}{7},\frac{0.4}{7},\frac{0.4}{7},\frac{0.4}{7},\frac{0.4}{7},\frac{0.4}{7},\frac{0.2}{3},\frac{0.2}{3},\frac{0.2}{3}\right).$$
The choice of risk budgets depends on the portfolio manager. In our case, we assign 40\% of risk budgets to bonds and equities, and 20\% to commodities. We then allocate these budgets equally within each class. The strategy corresponds to finding the RB portfolio for Expected Shortfall with a threshold $\alpha= 95\%$ by solving Problem~(\ref{stocRB_ES}) using daily returns in the look-back window and rebalancing the portfolio on the last business day of each week.

\begin{figure}[!ht]
\centering
     \includegraphics[width=1\textwidth]{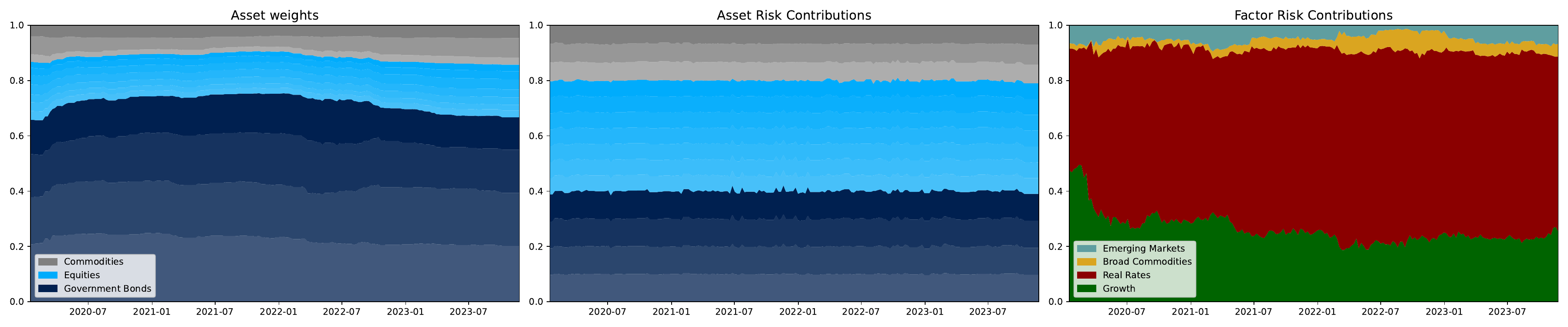}
      \caption{Evolution of asset weights and asset / factor risk contributions (in \%) for the Risk Budgeting strategy}
       \label{rb_contribs}
\end{figure}
\vspace{0mm} 

In Figure~\ref{rb_contribs}, we observe that the asset weights are stable over time and that the risk contributions of the assets are very well aligned with the given vector of risk budgets, with the exception of small deviations which are linked to the accuracy of the Expected Shortfall estimates at a relatively high threshold $\alpha$ on a sample with a limited number of points. However, when we examine the factor risk contributions, which the RB portfolio does not take into account when optimizing portfolios, we observe that the risk contributions of Real Rates can reach levels of the order of 70\%, which should preferably be avoided. On the contrary, the risk contributions from Broad Commodities are often marginal, despite the fact that 20\% of risk budgets have been allocated to the commodity asset class through asset risk budgets.  

\medskip

To address this issue, we can consider the long-only FRB strategy that allows us to control factor risk contributions. We choose the following factor risk budgets for Growth, Real Rates, Broad Commodities and Emerging Markets respectively:  $$ b_f = (0.35, 0.35, 0.20, 0.10).$$

\medskip

The long-only FRB strategy is expected to allocate risk across the factors very closely to the selected factor risk budgets (although not perfectly in line, as weights are constrained to be positive) keeping factor risk contributions at desired levels. As shown in Figure~\ref{frb_contribs}, this expectation is validated, with factor risk contributions of the portfolio consistently aligning with the chosen factor risk budgets. However, in this scenario, we lose control over asset risk contributions, resulting in highly fluctuating asset weights. This outcome is undesirable for an investment strategy as it leads to high turnover and consequently increased costs.

\begin{figure}[!ht]
\centering
     \includegraphics[width=\textwidth]{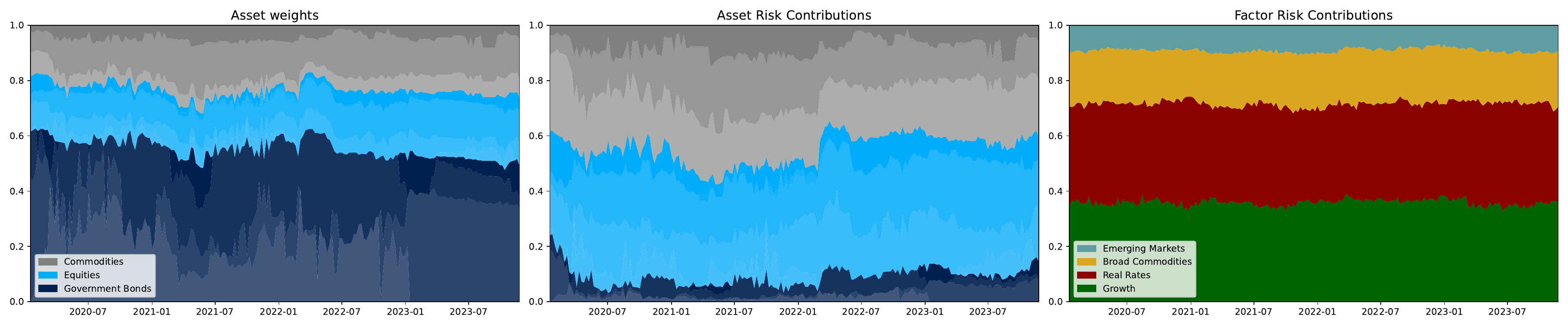}
      \caption{Evolution of asset weights and asset / factor risk contributions (in \%) for the long-only Factor Risk Budgeting strategy.}
       \label{frb_contribs}
\end{figure}
\vspace{3mm} 

The Asset-Factor Risk Budgeting strategy can bring solutions to the issues oberserved with the two above strategies. We consider the same asset and factor risk budgets, $b_a$ and $b_f$. We also need to set the asset and factor importance parameters, which we choose equal: $\lambda_a=0.5$ and $\lambda_f=0.5$. We can now compute the AFRB portfolio by solving Problem~(\ref{asset_factor_rb_stochastic}) for Expected Shortfall with a threshold $\alpha= 95\%$ on a weekly basis and implement the AFRB strategy.

\begin{figure}[!ht]
\centering
     \includegraphics[width=\textwidth]{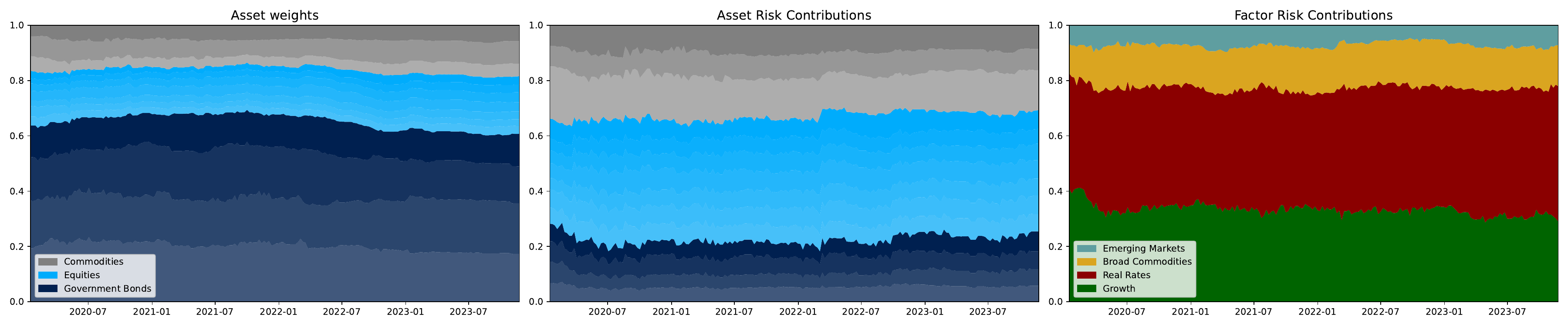}
      \caption{Evolution of asset weights and asset / factor risk contributions (in \%) for the Asset-Factor Risk Budgeting strategy.}
       \label{afrb_contribs}
\end{figure}
\vspace{3mm}

In Figure~\ref{afrb_contribs}, we observe that the AFRB strategy maintains a balance between asset and factor risk contributions. Compared with the case of RB, we observe a deviation from asset risk budgets. However, this is the price we have to pay to stay close to factor risk budgets. We note that the risk contribution of Broad Commodities is no longer small as it was with the RB strategy, and that we achieve a significant exposure to this factor while preserving the stability of asset weights contrary to the case of the long-only FRB. Growth and Real Rates now have balanced contributions, in line with factor risk budgets. Thanks to the importance parameters, we can choose to be more conservative in one of the two types of risk contributions, depending on the investment objective or market conditions. In short, the AFRB strategy can provide a flexible framework for better risk management of investment portfolios. 

\medskip

\begin{figure}[!ht]
\centering
     \includegraphics[width=.67\textwidth]{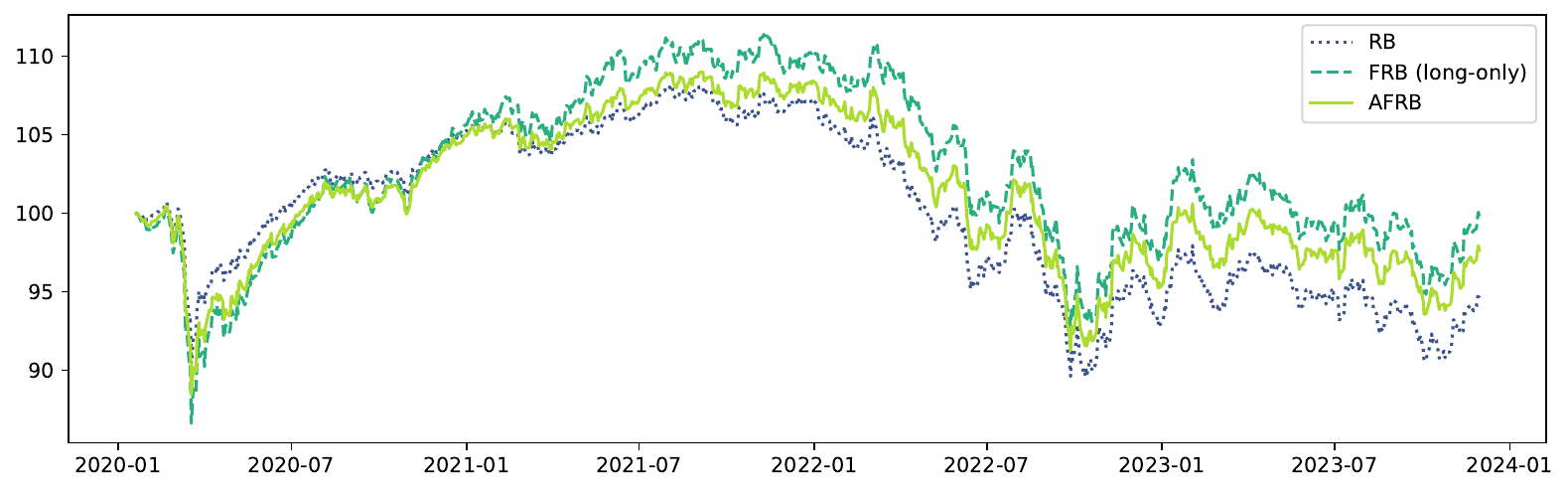}
      \caption{Cumulative performances of different strategies}
       \label{cum_perfs}
\end{figure}
\vspace{0mm} 

Although the primary objective of these strategies is diversification  rather than return maximization, it is still important to understand their behavior and performance beyond pure risk allocation figures. However, due to the limited length of our dataset, we can only analyse the strategies over a period of nearly four years which may not be sufficient to draw extensive conclusions, all the more given the major events over the period: the Covid-19 outbreak, the major increases in interest rates by central banks, periods of intense geopolitical tensions leading to war, etc. 

\begin{center}
\begin{tabular}{lrrrrr}
\toprule
{} &  $\hat{\mu}$ &  $\hat{\sigma}$ &  $\hat{\text{ES}}_{0.95}$ &  MaxDD &  Turnover \\
\midrule
RB              &        -1.42 &        6.55 &                0.93 &          17.12 &           1.00 \\
FRB (long-only) &        -0.07 &        7.99 &                1.18 &          16.71 &          17.25 \\
AFRB            &        -0.62 &        7.04 &                1.01 &          16.26 &           2.43 \\
\bottomrule
\end{tabular}
\captionof{table}{Sample mean (annualized), volatility (annualized), Expected Shortfall (daily), maximum drawdown and average turnover of the strategies (in \%)}
\label{stats}
\end{center}

Table~\ref{stats} demonstrates the main statistics associated with the different strategies computed on daily returns after transaction costs.\footnote{We account for transaction costs by factoring in bid/ask spreads of 0.02\% for equities, 0.01\% for bonds, and 0.05\% for commodities.} We observe negative annualized returns from the three strategies, primarily due to their high exposure to government bonds asset class and real rates factor, which performed poorly after the end of 2021 due to the increases in interest rates. However, the long-only FRB and AFRB strategies mitigate these losses by avoiding high exposures to this factor and also by providing significant exposure to the commodity factor, a factor that performed well after the Covid-19 outbreak. In terms of management costs, the long-only FRB strategy seems inefficient, despite the fact that it is the best performer, due to its high turnover, necessitating significant changes in positions every week and resulting in high transaction costs, which, in total, amounted to 0.31\% in our backtest (with 0.02\% for the RB and 0.05\% for the AFRB).\footnote{It is important to note that the results
may vary under other transaction cost schemes or rebalancing frequencies.} Conversely, the AFRB strategy is a low turnover strategy similar to the RB strategy limiting the implementation cost as well as allowing for a control on exposures to the desired factors.

\section{Conclusion}
In this paper, we address the question of incorporating factor risk contributions into the portfolio optimization process. First, we recall the concept of Risk Budgeting, which takes into account the risk contributions of individual assets. Next, we develop the necessary tools and propose a factor risk measure that allows us to apply these ideas to underlying risk factors in the presence of a linear factor model. We then formally define Factor Risk Budgeting and an associated stochastic optimization problem for finding FRB portfolios for a large set of risk measures. Finally, we propose a framework that combines RB and FRB to obtain portfolios that control the risk induced by both assets and factors. We conclude the paper with two case studies where we construct portfolios based on our ideas using equity and macroeconomic factor models.

\section*{Acknowledgment}
The authors would like to thank the participants of QuantMinds International 2023 held in London, the 68th Euro Working Group for Commodity and Financial Modelling (EWGCFM) conference held at Khalifa University in Abu Dhabi, and the XXV Workshop on Quantitative Finance held in Bologna for their valuable contributions to the discussions that occurred during these events. The quantitative teams at Amundi should also be thanked for their support. We would also like to express our gratitude to the two anonymous referees whose comments helped us to improve our article. 

\section*{Data Availability Statement}
The data that support the findings of this study are available from the corresponding author upon reasonable request.

\bibliographystyle{plain}

\end{document}